\documentclass[12pt,draftcls,onecolumn]{IEEEtran} 

\usepackage{graphicx} 
\usepackage{amsmath,amsthm}
\usepackage{bbm}
\usepackage{graphicx}  
\usepackage{subfigure} 
\usepackage{amsmath}
\usepackage{amssymb}
\usepackage{amstext}
\usepackage{stfloats}
\usepackage{color}
\graphicspath{{figures/}}

\usepackage[parfill]{parskip} 

\newcommand{\etal}{\emph{et~al.}}

\newcommand{\bm}[1]{\boldsymbol{#1}}
\newcommand{\mat}[1]{\bm{#1}}
\newcommand{\EE}{\mathbb{E}}

\newcommand{\lambertW}{\mathsf{W}}

\newtheorem{theorem}{Theorem}
\newtheorem{definition}{Definition}

\newtheorem{lemma}{Lemma}
\newtheorem{corollary}{Corollary}

\title{A Finite Block Length Achievability Bound for Low Probability of Detection Communication}

\author{Nick~Letzepis\thanks{N. Letzepis is with the Cyber and Electronic Warfare Division, Defence Science and Technology Group, West Avenue, Edinburgh SA 5111, Adelaide, Australia. {Email: {\tt nick.letzepis@ieee.org}.}},~\IEEEmembership{Member,~IEEE.}}

\begin{document}
\maketitle

%
%
\begin{abstract}
Low probability of detection (or covert) communication refers to the scenario where information must be sent reliably to a receiver, but with low probability of detection by an adversary.  Recent works on the fundamental limits of this communication problem have established achievability and converse bounds that are asymptotic in the block length of the code.  This paper uses Gallager's random coding bound to derive a new achievability bound that is applicable to low probability of detection communication in the finite block length regime.  Further insights are unveiled that are otherwise hidden in previous asymptotic analyses.
\end{abstract}

%
%
\section{Introduction}

Often in defence and national security, information must be conveyed
with \emph{low probability of detection} (LPD) by unauthorised
adversaries. The problem is similar to \emph{secure communications},
but stricter in the sense that the sender cannot afford the
transmission to be detected let alone its message being
compromised. In an information theoretic setting, the problem becomes
one of maximising the amount of information that can be sent reliably,
whilst satisfying a constraint on the probability of detection.  A
number of authors have studied LPD communication in this
context~\cite{Hero2003p3235,Bash2013p1921,Wang2015isit,Wang2016p3493,Bloch2016p2334}.

In~\cite{Hero2003p3235}, Hero studied the LPD problem in the context
of space-time codes for multi-antenna communication subjected to
quasi-static fading. Rather than constraining the probability of
detection, Hero maximises the information rate whilst constraining the
\emph{{C}hernoff information}~\cite[Theorem~11.9.1]{Cover2006book},
i.e. the best achievable exponent in the adversary's Bayesian
probability of error.  A salient point here is that Hero is only
constraining the rate at which detection error probability decays to
zero with the code word length (or block length), not the detection
error probability itself.  It turns out that for additive white Gaussian noise (AWGN) channels
the Shannon capacity of the channel is zero for a constraint
explicitly on the detection error probability sum, i.e. the sum of the adversary's
false and miss detection error probabilities. Moreover, in this
instance LPD communication obeys the so called \emph{square-root law} (SRL),
i.e. only $\mathcal{O}(\sqrt{n})$ bits of information can be
transmitted reliably in $n$ channel usages whilst constraining the
detection error probability sum~\cite{Bash2013p1921}.  Recently, Wang~\etal~\cite{Wang2015isit,Wang2016p3493} proved that the SRL extends to LPD communication over a broad class of discrete
memoryless channels as well.  Moreover, they showed that the rate of
increase of information with $\sqrt{n}$ is proportional to mutual
information maximised over all input distributions subject to a
constraint on the relative entropy.  It should be noted, however, that
their results assume the adversary observes the same channel
outputs as the intended receiver, which is unrealistic in many
practical situations.  In~\cite{Bloch2016p2334}, Bloch considers the LPD problem from a resolvability perspective for the more general case when both the receiver and adversary's channels are separate discrete memoryless channels (i.e. they do not observe the same channel outputs).  Not only does Bloch prove the SRL in this more general setting, but also fundamental limits on the asymptotic scaling of the message and key size when communication must be both covert and secret.  Central to the achievability results of~\cite{Wang2015isit,Wang2016p3493,Bloch2016p2334} is the use of a low weight, or \emph{sparse signalling} scheme (as referred to in this paper) to satisfy the LPD constraint.  This is where the probability of sending an \emph{innocent symbol} (defined as the channel input when no communication takes place~\cite{Bloch2016p2334}) approaches $1$ as $n\rightarrow \infty$.

This paper takes a different approach to~\cite{Bash2013p1921,Wang2015isit,Wang2016p3493,Bloch2016p2334} to establish an achievability bound that applies to finite block length codes.  The contributions of this paper are summarised as follows.  Following important preliminary details in Section~\ref{sec:prelim}, Section~\ref{sec:prob_det} revisits the implications of the LPD constraint on the input signalling density.  It is shown that the LPD constraint can be recast as a constraint on the chi-squared distance between the densities of the adversary's observations conditioned on transmission and no transmission, similar to~\cite{Bash2013p1921,Wang2015isit,Wang2016p3493,Bloch2016p2334}, but without using Pinsker's inequality~\cite[Theorem~2.33]{Yeung2002book}, Taylor series expansions or bounds on the natural logarithm.  Using this relationship the \emph{spareness factor}, defined as the probability of sending a non-innocent symbol, is derived in terms of the block length, LPD constraint and chi-squared distance.  In Section~\ref{sec:ach_lpd_info}, Gallager's error exponent~\cite{Gallager1965p3,Gallager1968book} is lower bounded in terms of the sparseness factor and the exponent of the density of the non-innocent symbols.  Combining this result with the constraint on the sparseness factor, a finite block length lower bound is derived on the number of bits that can be transmitted with non-vanishing LPD and decoding error probability.  Interestingly, in the finite block length regime the bound indicates there is an optimal block length that maximises the achievable information rate (bits per channel use), i.e. the achievable rate increases with $n$ until it reaches this optimal block length and then begins decreasing at a rate proportional to $1/\sqrt{n}$ as a result of the SRL.  In the asymptotic large block length regime, the lower bound can be written in terms of the mutual information of the Rx's channel similar to~\cite{Wang2015isit,Wang2016p3493,Bloch2016p2334}.  In Section~\ref{sec:lpd_examples}, the utility of the finite block length achievability bound is demonstrated for the well known binary symmetric channel (BSC) and AWGN channel~\cite{Cover2006book}.

%
%
\section{Preliminaries} \label{sec:prelim}

\begin{figure}[t]
\centering
\includegraphics[width=0.7\columnwidth]{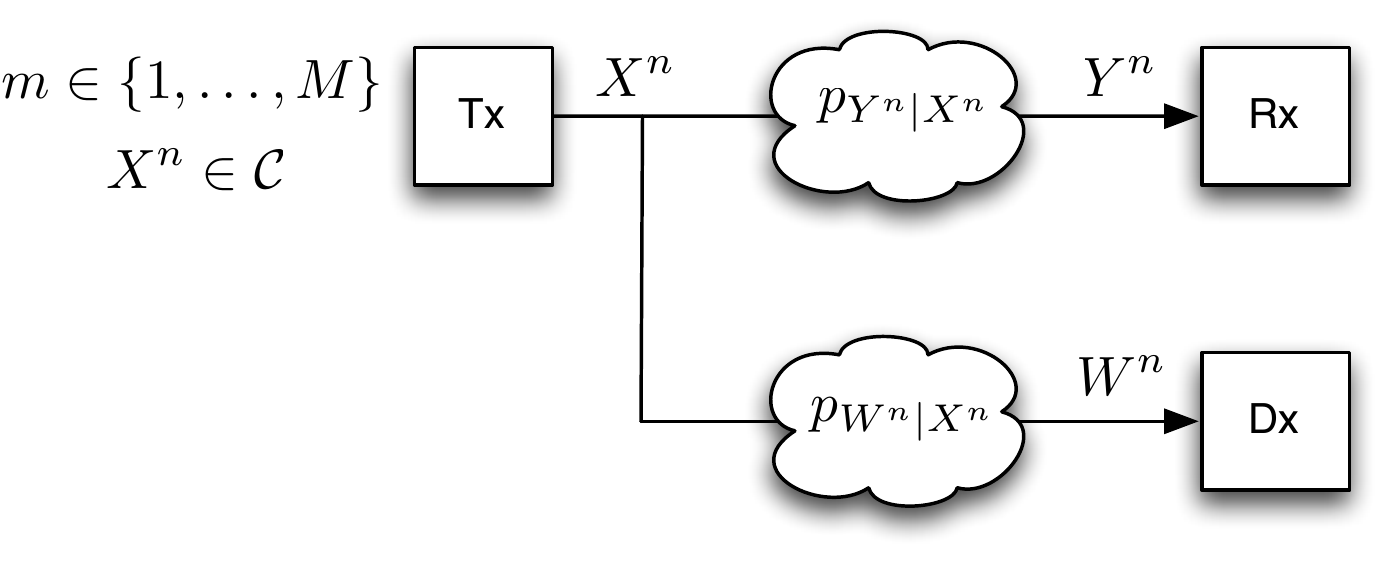}
\caption{The LPD communication scenario: the transmitter (Tx) sends
  message $m \in \{1, \ldots, M \}$ to the receiver (Rx) with LPD by
  an adversary with a detector (Dx).}
\label{fig:lpd_com_scenario}
\end{figure}

Consider the LPD communication scenario shown in
Fig.~\ref{fig:lpd_com_scenario} consisting of a transmitter (Tx), receiver (Rx) and detector (Dx).\footnote{In related works~\cite{Bash2013p1921,Wang2015isit,Wang2016p3493,Bloch2016p2334}, the Tx, Rx and Dx are referred to as "Alice", "Bob" and "Willie" respectively.} In this
scenario, the Tx wishes to send a message to the Rx with a LPD by the
Dx. Let $\mathcal{M} = (1,\ldots,M)$ denote the set of possible
messages the Tx can send.  For each message, the Tx constructs a
code book $\mathcal{C}$ consisting of $M$, $n$-length code words, where each code word, $\mat{x}_m = (x_1(m), \ldots, x_n(m))$, is generated independently
according to the joint density $p_{X^n}$ with support $\mathcal{X}^n$.
Further assume the symbols (or letters) of each code word are
independently and identically distributed (i.i.d.) so that $p_{X^n}
(\mat{x}) = \prod_{i=1}^{n} p_{X}(x_{i})$, $x_{i} \in \mathcal{X}$. To send message $m$, the Tx transmits code word $\mat{x}(m)$ to the Rx
who observes the corrupted version $\mat{y} = (y_1, \ldots, y_n) \in
\mathcal{Y}^n$ with probability $p_{Y^n | X^n} (\mat{y} | \mat{x}(m) )$. Given
knowledge of $\mathcal{C}$, the Rx decodes its observation outputting
the decision $\hat{m}$ and a decoding error occurs when $\hat{m} \neq m$.  When no message is transmitted, it is assumed
the Tx inputs $\mat{x}_0 = (x_0, \ldots, x_0) \in \mathcal{X}^n$ to the channel.
Using the terminology of~\cite{Bloch2016p2334}, $x_0$ is referred to as the innocent symbol, representing the case where the Tx expends no resources. 

Unfortunately for the Tx, the Dx also receives a corrupted version of
the code word, $\mat{w} = (w_1, \ldots, w_n) \in \mathcal{W}^n$ with
probability $p_{W^n|X^n}(\mat{w}|\mat{x}(m))$. The Dx does not know
$\mathcal{C}$, but knows the distribution used to construct it,
i.e. $p_{X}$.  Given this knowledge, the Dx attempts to decide whether
the input to its channel was either $\mat{x}_0$ or $\mat{x} \in
\mathcal{C}$.  It is assumed the Dx and Rx channels are memoryless and
conditionally independent, so that
$p_{Y^nW^n|X^n}(\mat{y},\mat{w}|\mat{x}) = \prod_{i=1}^{n}
p_{Y|X}(y_i|x_i) p_{W|X}(w_i|x_i)$.

As in~\cite{Bloch2016p2334}, it is assumed that $p_{W}$ is dominated by $p_{W|X=x_0}$ to exclude scenarios where the Dx would always detect the Tx's code word with non-vanishing probability, or never detect it.\footnote{In other words, let $\mathcal{W}_0 = \{w: p_{W|X}(w|x_0) = 0\}$, if $p_{W}$ is dominated by $p_{W|X=x_0}$ then $p_{W}(w) = 0$ for all $w\in \mathcal{W}_0$.}  Similarly it is assumed $p_{Y}$ is dominated by $p_{Y|X=x_0}$ to ensure the Rx does not have an unfair advantage over the Dx.

%
%
%

\section{Detection Error Probability} \label{sec:prob_det}

From its observation $\mat{w}$ the Dx must decide between two
hypotheses: $\mathcal{H}_0$, the input to its channel was $\mat{x}_0$;
and $\mathcal{H}_1$, the input to its channel was $\mat{x} \in
\mathcal{C}$. The probability of the Dx's observation conditioned on
$\mathcal{H}_0$ and $\mathcal{H}_1$ are given by
\begin{align}
\Pr \{ \mat{w}| \mathcal{H}_0 \} &= p_{W^n|X^n}(\mat{w} | \mat{x}_0) = \prod_{i=1}^{n} p_{W|X}(w_i|x_0) \label{eq:prob_w_h0}\\
\Pr \{ \mat{w}| \mathcal{H}_1 \} &= p_{W^n}(\mat{w}) = \prod_{i=1}^{n} \int_{\mathcal{X}} p_{X}(x) p_{W|X}(w_i|x) \, dx, \label{eq:prob_w_h1}
\end{align}
where~\eqref{eq:prob_w_h1} assumes the messages are equally
likely. Let $\alpha$ and $\beta$ denote the false and miss detection
probabilities respectively, i.e. $\alpha$ is the probability of
deciding $\mathcal{H}_1$ when $\mathcal{H}_0$ is true, and $\beta$ is
the probability of deciding $\mathcal{H}_0$ when $\mathcal{H}_1$ is
true. Further define $\alpha + \beta$ as the detection error
probability sum. Assume the Dx knows the statistics of its
observations and employs an optimal statistical
hypothesis test that minimises $\alpha + \beta$.  From~\cite[Theorem~13.1.1]{Lehmann2005book} for any such optimal test
\begin{equation}
\alpha + \beta = 1 - d_{\rm TV}(p_{W^n}, p_{W^n|X^n=\mat{x}_0}) \label{eq:det_err_sum}
\end{equation}
where $d_{\rm TV}(p,q) = \frac{1}{2} \int_{\mathcal{X}} |p(x) - q(x) |\, dx$ is the variational distance between densities $p(x)$ and $q(x)$ of support $\mathcal{X}$.  

The goal of the Tx is to force $\alpha + \beta$ to be close to $1$ so that the adversary's best statistical test is not much better than a blind one (i.e. a test that ignores the channel observation $\mat{w}$)~\cite{Bloch2016p2334}.  Toward this end, suppose the Tx must ensure $\alpha + \beta \geq 1 - \epsilon_{\rm det}$, where $0< \epsilon_{\rm det} \lll 1$ is close to zero.  Then from~\eqref{eq:det_err_sum}, the Tx's code book design must be constrained such that
\begin{equation}
d_{\rm TV} (p_{W^n}, p_{W^n|X^n=\mat{x}_0}) \leq \epsilon_{\rm det}. \label{eq:vardist_constraint}
\end{equation} %
In this paper, the constraint in~\eqref{eq:vardist_constraint} is referred to as a \emph{low probability of detection} constraint, i.e. a constraint on the adversary's probability of detection $\epsilon_{\rm det}$. Unfortunately, dealing with the total variation distance between two multivariate densities is problematic.  Instead, the works of~\cite{Bash2013p1921,Wang2016p3493,Bloch2016p2334} proceed to upper bound $d_{\rm TV} (p_{W^n}, p_{W^n|X^n=\mat{x}_0})$ in terms of the relative entropy via Pinsker's inequality~\cite[Theorem~2.33]{Yeung2002book} and then further approximate or weaken the bound using Taylor series expansion methods.  In particular, using bounds on the natural logarithm, Bloch~\cite{Bloch2016p2334} bounds the relative entropy in terms of the chi-squared distance plus other related terms.  A simpler approach is presented in the following lemma.
\begin{lemma} \label{lem:var_dist_ineq}
Define $p_{W^n}$ and $p_{W^n|X^n}$ as in~\eqref{eq:prob_w_h0} and~\eqref{eq:prob_w_h1} respectively. Then,
\begin{equation}
d_{\rm TV}(p_{W^n}, p_{W^n|X^n=\mat{x}_0}) \leq \frac{1}{2} \sqrt{ \lambertW_0^{-1}\left(  n \chi^2( p_{W} \parallel p_{W|X=x_0}) \right) } \label{eq:var_dist_ineq}
\end{equation}
where $\lambertW_{0}^{-1}(z) = ze^{z}$, $z \geq -1/e$, is the inverse of the Lambert-$\lambertW$ function on the principle branch~\cite{Corless1996} and $\chi^2(p \parallel q) \triangleq \int_{\mathcal{X}} \frac{(p(x) - q(x))^2}{q(x)} \, dx$ is the chi-squared distance.
\end{lemma}
\begin{proof}
By the definition of variational distance,
\begin{align}
d_{\rm TV}(p_{W^n}, p_{W^n|X^n=\mat{x}_0}) &= \frac{1}{2} \int_{\mathcal{W}^n} \left| p_{W^n}(\mat{w}) -  p_{W^n|X^n}(\mat{w} | \mat{x}_0) \right| \, d \mat{w} \notag \\
&= \frac{1}{2} \int_{\mathcal{W}^n} \frac{\left| p_{W^n}(\mat{w}) -  p_{W^n|X^n}(\mat{w} | \mat{x}_0) \right|}{\sqrt{p_{W^n|X^n}(\mat{w} | \mat{x}_0) }} \sqrt{p_{W^n|X^n}(\mat{w} | \mat{x}_0) }   \, d \mat{w} \notag \\
&\leq \frac{1}{2}  \sqrt{ \int_{\mathcal{W}^n} \frac{\left( p_{W^n}(\mat{w}) -  p_{W^n|X^n}(\mat{w} | \mat{x}_0) \right)^2}{p_{W^n|X^n}(\mat{w} | \mat{x}_0) } \, d\mat{w} }  \sqrt{ \int_{\mathcal{W}^n} p_{W^n|X^n}(\mat{w} | \mat{x}_0) \, d\mat{w}}   \notag \\
&= \frac{1}{2} \sqrt{ \chi^2( p_{W^n} \parallel p_{W^n|X^n=\mat{x}_0})}, \label{eq:vardist_xidist_work}
\end{align}
where the second line follows since $p_{W}$ is dominated by $p_{W|X=x_0}$ by assumption and the third line results from application of H\"{o}lder's inequality.

Again, since $p_{W}$ is dominated by $p_{W|X=x_0}$ by assumption, application of \cite[Lemma~3.3.10]{reiss1989approximate} yeilds
\begin{align}
\chi^2( p_{W^n} \parallel p_{W^n|X^n=\mat{x}_0}) &\leq  n \chi^2( p_{W} \parallel p_{W|X=x_0})   e^{ n  \chi^2( p_{W} \parallel p_{W|X=x_0}) } \notag \\
&= \lambertW_0^{-1} \left( n \chi^2( p_{W} \parallel p_{W|X=x_0}) \right). \label{eq:reiss_lemma}
\end{align}
Thus combining~\eqref{eq:vardist_xidist_work} and~\eqref{eq:reiss_lemma} results in~\eqref{eq:var_dist_ineq} as stated in the lemma. 
\end{proof}
Application of Lemma~\ref{lem:var_dist_ineq} results in the following corollary.
\begin{corollary}
Suppose the Tx must constrain $\alpha + \beta \geq 1 - \epsilon_{\rm det}$ for $0 < \epsilon_{\rm det} < 1$. This constraint can be satisfied by ensuring
\begin{equation}
\chi^2(p_{W} \parallel p_{W|X=x_0}) \leq \frac{4 \xi^2 \epsilon_{\rm det}^2}{n} = \frac{1}{n} \left( 4 \epsilon^2_{\rm det} + \mathcal{O}(\epsilon^4_{\rm det}) \right) \label{eq:chi2_constraint}
\end{equation}
where $\xi = e^{-\frac{1}{2} \lambertW_0( 4 \epsilon_{\rm det}^2 )}$.
\end{corollary}
\begin{proof}
From~\eqref{eq:det_err_sum}, to constrain $\alpha + \beta \geq 1 - \epsilon_{\rm det}$ requires $d_{\rm TV}(p_{W^n}, p_{W^n|X^n=\mat{x}_0}) \leq \epsilon_{\rm det}$. Thus using Lemma~\ref{lem:var_dist_ineq} this can be satisfied by ensuring
\begin{align}
\frac{1}{2} \sqrt{ \lambertW_0^{-1}\left(  n \chi^2( p_{W} \parallel p_{W|X=x_0}) \right) } \leq \epsilon_{\rm det}.
\end{align}
Since $\lambertW_0(x)$ is a monotonically increasing function of $x$ for $x > -1/e$ then one may write
\begin{align}
\chi^2( p_{W} \parallel p_{W|X=x_0})   &\leq \frac{1}{n}\lambertW_0 \left( 4\epsilon_{\rm det}^2 \right) \notag \\
&= \frac{4}{n} \epsilon^2_{\rm det} e^{- \lambertW_0 \left( 4\epsilon_{\rm det}^2 \right)}
\end{align}
\end{proof}

To avoid detection, it is clear from the above analysis that the Tx should use a code book whose code words are close to $\mat{x}_0$.  One approach, in a similar vein as~\cite{Wang2015isit,Wang2016p3493,Bloch2016p2334}, is to use a sparse signalling scheme, i.e. the Tx generates
code words such that symbols $x \neq x_0$ occur very infrequently. More formally, in this paper the sparse signalling density
is defined as follows.
\begin{definition}[Sparse signalling density] \label{def:sparse_sig}
Let $p_{\tilde{X}}(x)$ denote an arbitrary kernel density. The sparse
signalling density is defined as
\begin{equation}
p_{X}(x;\tau, p_{\tilde{X}}) = (1-\tau) \delta(x-x_0) + \tau p_{\tilde{X}} (x), \label{eq:sparse_signalling_density}
\end{equation}
where $0 \leq \tau \leq 1$ is the sparseness factor and $\delta(x)$ denotes the Dirac-delta function.
\end{definition}
Substituting~\eqref{eq:sparse_signalling_density} into~\eqref{eq:chi2_constraint} translates the detection error
probability sum constraint to a constraint on the sparseness factor, i.e. the constraint $\alpha + \beta \geq 1 - \epsilon_{\rm det}$ can be satisfied using a sparse signalling density with
\begin{equation}
\tau \leq   \frac{ 2 \xi \epsilon_{\rm det} }{ \sqrt{n \chi^2( p_{\tilde{W}} \parallel p_{W|X=x_0})}}, \label{eq:tau_constraint}
\end{equation}
where $p_{\tilde{W}}(w) = \int_{\mathcal{X}} p_{\tilde{X}}(x) p_{W|X}(w|x) \, dx$ is the density of the Dx's channel outputs induced by the Kernel density $p_{\tilde{X}}$.

%
%
\section{LPD Achievability Bound} \label{sec:ach_lpd_info}
The achievability results of previous works~\cite{Bash2013p1921,Wang2015isit,Wang2016p3493,Bloch2016p2334} on LPD communication are only applicable in the asymptotic large block length regime.  The approach of Wang~\etal~\cite{Wang2015isit,Wang2016p3493} is based on one-shot achievability bounds~\cite{Wang2009isit,Polyanskiy2010p2307}.  While, Bloch's approach~\cite{Bloch2016p2334}, uses suitably modified typical sets to enable the application of concentration inequalities.  This section takes a much simpler approach to the aforementioned works, which not only proves the asymptotic achievability of the SRL, but also yields additional insights on the number of achievable bits in the finite block length regime. Central to this approach is Gallager's coding theorem~\cite{Gallager1965p3,Gallager1968book} stated as follows.
\begin{theorem} [Gallager's coding theorem~\cite{Gallager1965p3,Gallager1968book}] \label{thrm:gallager}
For any $p_{X^n}(\mat{x}) = \prod_{i=1}^{n} p_{X}(x_i)$ and $0 \leq
\rho \leq 1$ there exists a code with $M$, $n$-length code words with
average probability of decoding error, $\epsilon_{\rm dec} > 0$, such
that
\begin{align}
\epsilon_{\rm dec} \leq  \exp\left( -n \left[  \mathcal{E}_{0}(\rho,p_{X}, p_{Y|X}) -  \frac{\rho}{n}\log M \right] \right), \label{eq:gallager_iid_memoryless}
\end{align}
where
\begin{align}
 & \mathcal{E}_{0}(\rho,p_{X}, p_{Y|X})  = -\log \int_{\mathcal{Y}} \left\{ \int_{\mathcal{X}} p_{X}(x) \left[ p_{Y |X}(y|x) \right]^{\frac{1}{1+\rho}} \, d x \right\}^{1+\rho} \, d y. \label{eq:gal_exponent}
\end{align}
\end{theorem}
The maximisation of the square-bracketed term in~\eqref{eq:gallager_iid_memoryless} over $p_X$ and $\rho$ is Gallager's well known \emph{error exponent} and describes how rapid the average probability of decoding error reduces with increasing block length~\cite{Gallager1965p3,Gallager1968book}.  When sparse signalling is employed,~\eqref{eq:gal_exponent} is lower bounded as follows.
\begin{lemma} \label{lem:sparse_exponent_minkowski}
Suppose the input density is sparse as defined by~\eqref{eq:sparse_signalling_density} with $0 \leq \tau < 1$, then for $0 \leq \rho \leq 1$,
\begin{align}
\mathcal{E}_{0} (\rho,p_{X}, p_{Y|X}) &\geq - (1+\rho) \log \left[   1-\tau \left(1 -  e^{- \frac{1}{1+\rho} \mathcal{E}_0(\rho,p_{\tilde{X}},p_{Y|X})} \right)  \right]  \label{eq:E0_tau_log} \\
&\geq (1+\rho) \tau \left(1 -  e^{- \frac{1}{1+\rho} \mathcal{E}_0(\rho,p_{\tilde{X}},p_{Y|X})} \right) \label{eq:E0_tau}
\end{align}
\end{lemma}
\begin{proof}
See Appendix~\ref{app:sparse_exponent_minkowski_proof}.
\end{proof}
From Lemma~\ref{lem:sparse_exponent_minkowski}, it can be seen that the lower bound~\eqref{eq:E0_tau} is proportional the sparseness factor $\tau$ times a factor that is dependent on the $\mathcal{E}_0$ value of the underlying kernel density.  Moreover, when $\tau = 1$,~\eqref{eq:E0_tau_log} reverts to~\eqref{eq:gal_exponent}.  On the other hand if $\tau = 0$, or $ \mathcal{E}_0(\rho,p_{\tilde{X}},p_{Y|X}) = 0$ then the lower bound is also zero. Since~\eqref{eq:E0_tau} is proportional to the sparseness factor $\tau$ and from~\eqref{eq:tau_constraint} $\tau$ is $\mathcal{O}(1/\sqrt{n})$ to satisfy the LPD constraint, then the following theorem is proved.
\begin{theorem} \label{thrm:gallager_tau}
Suppose the Tx employs the sparse signalling scheme~\eqref{eq:sparse_signalling_density} with sparseness factor~\eqref{eq:tau_constraint}. Then there
exists a code with $M$, $n$-length code words with an average decoding
error probability not exceeding $\epsilon_{\rm dec}$ and probability
of detection not exceeding $\epsilon_{\rm det}$ such that for any $0
\leq \rho \leq 1$,
\begin{align}
& \log_2 M \geq   \epsilon_{\rm det} \sqrt{n}  L(\rho) + \frac{1}{\rho} \log_2 \epsilon_{\rm dec}. \label{eq:ach_covert_bits}
\end{align}
where
\begin{align}
L(\rho)   &=   \frac{2 \xi }{\log 2} \left(\frac{1+\rho}{\rho} \right)  \left( \frac{ 1 - e^{- \frac{1}{1+\rho} \mathcal{E}_0(\rho,p_{\tilde{X}},p_{Y|X})}}{\sqrt{ \chi^2(p_{\tilde{W}} \parallel p_{W|X=x_0}) }} \right)  \label{eq:L_def}
\end{align}
\end{theorem}
From~\eqref{eq:ach_covert_bits}, some interesting
insights can be made.  Firstly, for large block lengths the achievability bound is dominated by the first term of~\eqref{eq:ach_covert_bits}, and as expected from the analyses of~\cite{Bash2013p1921,Wang2015isit,Wang2016p3493,Bloch2016p2334}, scales asymptotically with $\mathcal{O}(\sqrt{n})$ as a consequence of the SRL. The quantity $L(\rho)$ is of a similar form as ~\cite[eq.~(28)]{Wang2016p3493}, but now takes into account both the Rx and Dx's channel statistics, whereas~\cite[eq.~(28)]{Wang2016p3493} assumes the Dx has the same observations as the Rx.  In the asymptotic large block length regime one has the following corollary.
\begin{corollary} \label{cor:L_asymp}
\begin{align}
\lim_{n \rightarrow \infty} & \frac{\log_2 M}{ \epsilon_{\rm det} \sqrt{n}}  \geq \frac{2 \xi }{\log 2} \frac{I(\tilde{X};Y)}{\sqrt{\chi^2(p_{\tilde{W}} \parallel p_{W|X=x_0})}} \label{eq:L_asymptotic}
\end{align}
where $I(X;Y) \triangleq \int_{\mathcal{X}} \int_{\mathcal{Y}} p_{XY}(x,y) \log \frac{p_{XY}(x,y)}{p_{X}(x) p_{Y}(y)} \, dy \, dx$ is the mutual information between random variables $X$ and $Y$~\cite[Sec~2.3]{Cover2006book}.
\end{corollary}
\begin{proof}
Dividing both sides of~\eqref{eq:ach_covert_bits} by $\sqrt{n}$, yields $\lim_{n\rightarrow \infty} \frac{1}{\sqrt{n}} \log_2 M \geq \epsilon_{\rm det} L(\rho)$. From~\cite[Theorem~2]{Gallager1965p3}, $\mathcal{E}_0(\rho,p_{\tilde{X}},p_{Y|X})$ is a non-decreasing function and $  \frac{\partial}{\partial \rho}  \left. \mathcal{E}_0(\rho,p_{\tilde{X}},p_{Y|X}) \right|_{\rho=0} = I(\tilde{X};Y) $.  Therefore $L(\rho)$ is maximised when $\rho = 0$ and hence~\eqref{eq:L_asymptotic} is obtained using L'Hospital's Rule~\cite[Sec.~3.4]{Abramowitz1972book} on~\eqref{eq:L_def}.
\end{proof}
From~\eqref{eq:L_asymptotic} it can be observed that replacing $\chi^2(p_{\tilde{W}} \parallel p_{W|X=x_0})$ with $\chi^2(p_{\tilde{Y}} \parallel p_{Y|X=x_0})$ and multiplying by $\xi (\log 2)/\sqrt{2} \approx (\log 2)/\sqrt{2}$ results in~\cite[eq.~(28)]{Wang2016p3493}.  The first alteration is required to ensure the Dx and Rx have the same channel as assumed in~\cite{Wang2016p3493}, and the $\log 2$ scaling is to convert from bits to nats.  The extra division by $\sqrt{2}$ is a consequence of~\cite{Wang2016p3493} constraining the relative entropy rather than the variational distance as done in this paper.  There is a further subtle difference between~\eqref{eq:L_asymptotic} and~\cite[eq.~(28)]{Wang2016p3493}, in that~\eqref{eq:L_asymptotic} is expressed directly in terms of $I(\tilde{X};Y)$ whereas~\cite[eq.~(28)]{Wang2016p3493} requires $\EE[ D(p_{Y|X=\tilde{X}} \parallel p_{Y|X=x_0}) ]$, a consequence of using the Taylor series expansion of $I(X;Y)$ in $\tau$.

In contrast to Corollary~\ref{cor:L_asymp}, for the finite block length regime the last term of~\eqref{eq:ach_covert_bits}  becomes significant, particularly if one considers the information rate as demonstrated in the following corollary.
\begin{corollary} \label{cor:optimal_n}
In Theorem~\ref{thrm:gallager_tau}, define the achievable information rate
\begin{equation}
R(n) \triangleq \frac{1}{n} \log_2 M \geq \frac{\epsilon_{\rm det}  L(\rho)}{\sqrt{n}} + \frac{1}{n\rho} \log_2 \epsilon_{\rm dec} \label{eq:ach_lpd_rate}
\end{equation}
for any $0 \leq \rho \leq 1$. Then~\eqref{eq:ach_lpd_rate} is maximised by
\begin{equation}
\sqrt{n^*} = \frac{2  \log_2 \frac{1}{\epsilon_{\rm dec}}}{ \epsilon_{\rm det} \rho L(\rho)} \label{eq:optimal_n}
\end{equation}
for which
\begin{equation}
R(n^*) \geq  \frac{\epsilon^2_{\rm det} }{4 \log_2 \frac{1}{\epsilon_{\rm dec}}} \rho L^2(\rho). \label{eq:optimal_R}
\end{equation}
Moreover, the number of achievable bits that can be sent in $n^*$ channel uses is given by
\begin{equation}
k^* = n^* R(n^*) = \frac{1}{\rho} \log_{2} \frac{1}{\epsilon_{\rm dec}}. \label{eq:k_star}
\end{equation}
\end{corollary}
Corollary~\ref{cor:optimal_n} shows that in the finite block length regime, the achievable rate first increases with  $n$ to a maximum value $R(n^*)$ and then begins decreasing as the SRL dominates.  From~\eqref{eq:optimal_n} it can be seen that $n^*$ increases with decreasing $\epsilon_{\rm det}$ and $\epsilon_{\rm dec}$, while on the other hand, from~\eqref{eq:optimal_R},  $R(n^*)$ decreases.

%
%
\section{LPD Communication Examples} \label{sec:lpd_examples}
To highlight the insights of the previous section, this section applies Theorem~\ref{thrm:gallager_tau} to two well known channels studied in the literature - the BSC and AWGN channel~\cite{Cover2006book}.

%
%
\subsection{Binary Symmetric Channel}
\begin{theorem} \label{thrm:bsc_ach_rate}
Suppose the Dx and Rx's channels are independent BSCs with crossover probabilities $\epsilon_{\rm Dx}$ and $\epsilon_{\rm Rx}$ respectively. Then for any $0 \leq \rho  \leq 1$,
\begin{equation}
\log_2 M \geq \sqrt{n} \epsilon_{\rm det} L_{\rm BSC}(\rho,\epsilon_{\rm Rx},\epsilon_{\rm Dx} )  + \frac{1}{\rho} \log_2 \epsilon_{\rm dec} \label{eq:bsc_ach_bits}
\end{equation}
where
\begin{align}
L_{\rm BSC}(\rho,\epsilon_{\rm Rx},\epsilon_{\rm Dx} ) = & \frac{2\xi }{\log 2} \left( \frac{1+\rho}{\rho} \right) \frac{\sqrt{\epsilon_{\rm Dx}(1 - \epsilon_{\rm Dx})}}{(1 - 2 \epsilon_{\rm Dx})}  \notag \\
& \times \left[ (1-\epsilon_{\rm Rx})^{\frac{1}{1+\rho}} - \epsilon_{\rm Rx}^{\frac{1}{1+\rho}} \right] \left[ (1-\epsilon_{\rm Rx})^{\frac{\rho}{1+\rho}} - \epsilon_{\rm Rx}^{\frac{\rho}{1+\rho}} \right]. \label{eq:L_bsc}
\end{align}
In addition,
\begin{align}
\lim_{n \rightarrow \infty} \frac{1}{\sqrt{n}} \log_2 M \geq &  2 \epsilon_{\rm det} \xi  \sqrt{\epsilon_{\rm Dx}(1 - \epsilon_{\rm Dx})} \frac{1 - 2 \epsilon_{\rm Rx}}{1 - 2 \epsilon_{\rm Dx}}  \log_2 \left( \frac{1 - \epsilon_{\rm Rx}}{\epsilon_{\rm Rx}} \right). \label{eq:bsc_asymp_L}
\end{align}
\end{theorem}
\begin{proof}
See Appendix~\ref{app:proof_bsc_ach_rate}
\end{proof}
From~\eqref{eq:L_bsc}, $L_{\rm BSC}(\rho, \epsilon_{\rm Rx}, 0) = L_{\rm BSC}(\rho, \epsilon_{\rm Rx},1) = 0$, i.e. when the Dx has a perfect (or perfectly inverted) channel, no information can be sent covertly.  On the other hand, $L_{\rm BSC}(\rho, \epsilon_{\rm Rx}, \frac{1}{2}) = \infty$, i.e. when Dx's channel is useless for
detection the Tx can send information at an order greater than
$\sqrt{n}$. In fact, for this case,
constraint~\eqref{eq:chi2_constraint} is redundant, since $\chi^2(p_{W} \parallel p_{W|X=x_0}) = 0$ regardless of $p_{X}$ and therefore the Tx can transmit unfettered
at the capacity of the Rx's channel.  If the Rx's channel is useless, $L_{\rm BSC}(\rho,\frac{1}{2}, \epsilon_{\rm Dx}) = 0$, as expected, no information can be sent. As shown in Appendix~\ref{app:proof_bsc_ach_rate} the lower bound~\eqref{eq:bsc_ach_bits} requires taking the limit $p_{\tilde{X}}(1) \rightarrow 0$.  This peculiarity does not imply that a binary $1$ is never transmitted, but rather results in the overall probability of sending a $1$, i.e. $\tau p_{\tilde{X}}(1)$, being $\mathcal{O}(1/\sqrt{n})$.

Fig.~\ref{fig:bsc_ach_rate} illustrates these insights for the case when $\epsilon_{\rm det} = 0.1$, $\epsilon_{\rm dec} = 10^{-3}$, and various $\epsilon_{\rm Dx}$. The solid lines plot the finite block length achievable information rate, i.e.~\eqref{eq:bsc_ach_bits} divided $n$, and exhibits the behaviour as expected from Corollary~\ref{cor:optimal_n}. For small $n$, the achievable rate increases with $n$ up to a maximum value and then decreases with $n$ as the SRL dominates.  The dashed lines plot the asymptotic large block length achievable rate, i.e.~\eqref{eq:bsc_asymp_L} divided by $\sqrt{n}/\epsilon_{\rm det}$, confirming the convergence of the finite block length results to the asymptotic large block length bound as $n$ increases.  Interestingly, this convergence is slow with $n$, which further highlights the importance of finite block length analysis in LPD communication system design.  

%
%
\subsection{AWGN Channel}
For the AWGN channel, applying Theorem~\ref{thrm:gallager_tau} yields the following result.
\begin{theorem} \label{thrm:awgn_ach_rate}
Suppose the Rx and Dx channels are independent AWGN channels with noise variance $\sigma_{\rm Rx}^2$ and $\sigma_{\rm Dx}^2$ respectively. Assuming a Gaussian kernel density, then for any $0 \leq \rho \leq 1$,
\begin{align}
\log_2 M \geq &  \frac{ \sqrt{2n} \epsilon_{\rm det}  \xi  }{(1+\rho)\log 2} \frac{\sigma_{\rm Dx}^2}{\sigma_{\rm Rx}^2}   + \frac{1}{\rho} \log_2 \epsilon_{\rm dec}. \label{eq:awgn_lpd}
\end{align}
Provided 
\begin{equation}
n > n_{\rm min} = \frac{1}{2 \xi^2  \epsilon_{\rm det}^2}  \frac{\sigma_{\rm Rx}^4}{\sigma_{\rm Dx}^4}  \log^2 \frac{1}{\epsilon_{\rm dec}} \label{eq:awgn_n_min}
\end{equation}
then~\eqref{eq:awgn_lpd} is maximised by
\begin{equation}
\rho^* = \frac{1}{\sqrt{\frac{n}{n_{\rm min}}} - 1} \left[ 1 + \left(\frac{n}{n_{\rm min}} \right)^{\frac{1}{4}} \right]. \label{eq:rho_opt_awgn}
\end{equation}
The peak achievable information rate and the block length it occurs at are given by
\begin{align}
R_{\rm awgn}(n^*) &= \frac{\epsilon^2_{\rm det} \xi^2 }{8 (\log 2) \log \frac{1}{\epsilon_{\rm dec}}} \frac{\sigma_{\rm Dx}^4}{\sigma_{\rm Rx}^4} \label{eq:R_star_awgn} \\
n^* &= \frac{8 \log^2 \epsilon_{\rm dec}}{\epsilon_{\rm det}^2 \xi^2 } \frac{\sigma_{\rm Rx}^4}{\sigma_{\rm Dx}^4} \label{eq:n_star_awgn}
\end{align}
respectively. Moreover, the number of achievable information bits that can be sent using $n^*$ channel uses is given by
\begin{equation}
k^* = n^* R_{\rm awgn}(n^*) = \log_2 \frac{1}{\epsilon_{\rm dec}}. \label{eq:k_star_awgn}
\end{equation}
\end{theorem}
\begin{proof}
See Appendix~\ref{app:proof_awgn_ach_rate}.
\end{proof}
From Theorem~\ref{thrm:awgn_ach_rate} it is clear that the ratio $\frac{\sigma^2_{\rm Dx}}{\sigma^2_{\rm Rx}}$ plays an important role in the achievability bound. Intuitively, the noisier the Dx's channel is relative to the Rx, then more bits are achievable with LPD and the smaller the minimum block length~\eqref{eq:awgn_n_min}.  Note that from~\eqref{eq:k_star_awgn}, while $k^*$ is constant, as $\frac{\sigma^2_{\rm Dx}}{\sigma^2_{\rm Rx}}$ decreases, the required number of channel uses~\eqref{eq:n_star_awgn} increases and hence the peak achievable information rate decreases.  These insights are illustrated in Fig.~\ref{fig:awgn_ach_rate} which plots~\eqref{eq:awgn_lpd} divided by $n$ (solid lines) compared to the asymptotic large block length result (dashed lines), i.e. the first term of~\eqref{eq:awgn_lpd} divided by $n$ with $\rho = 0$, for $\epsilon_{\rm det} = 0.1$, $\epsilon_{\rm dec} = 10^{-3}$ and various $\frac{\sigma_{\rm Dx}^2}{\sigma^2_{\rm Rx}}$. In addition, the minimum block length required to support a positive achievable information rate~\eqref{eq:awgn_n_min} is also shown (dot-dashed line).

%
%

\section{Conclusion} \label{sec:conc}

This paper studied reliable communication subject to a LPD constraint in the finite block length regime.  It was shown that the LPD constraint can be weakened to a constraint on the chi-squared distance, but without resorting to Taylor series expansions and bounds on the natural logarithm as required by other works in the literature~\cite{Bash2013p1921,Wang2015isit,Wang2016p3493,Bloch2016p2334}.  A new lower bound on the number of bits that can be transmitted both reliably and covertly was derived via Gallager's work on error exponents~\cite{Gallager1965p3,Gallager1968book}.  In particular, it was shown that using a sparse signalling scheme, the overall error exponent can be upper bounded in terms of the sparseness factor and the exponent of the underlying kernel density.  The resulting achievability bound exhibited behaviour whereby the number of achievable bits increases rapidly for small block lengths until it reaches a threshold  block length at which point the rate of increase is then dominated by the SRL, i.e. the number of achievable bits increases on the order of the square root of the block length.  In terms of the achievable information rate (bits per channel use), this corresponds to the achievable rate increasing with the number of channel uses until the threshold block length is reached and then begins decreasing at a rate that is inversely proportional the square root of the block length.  These results were applied to the BSC and AWGN channel.  For the BSC the achievability bound was derived in terms of the Rx and Dx's crossover probabilities.  For the AWGN channel, the resulting bound was shown to be directly proportional to the Dx-to-Rx noise power ratio.

\section*{Acknowledgement}
The author is thankful to Prof. Albert Guill\'{e}n i F\`{a}bregas for insightful discussions on the work presented in this paper.
\begin{figure}[t]
\centering
\includegraphics[width=0.9\columnwidth]{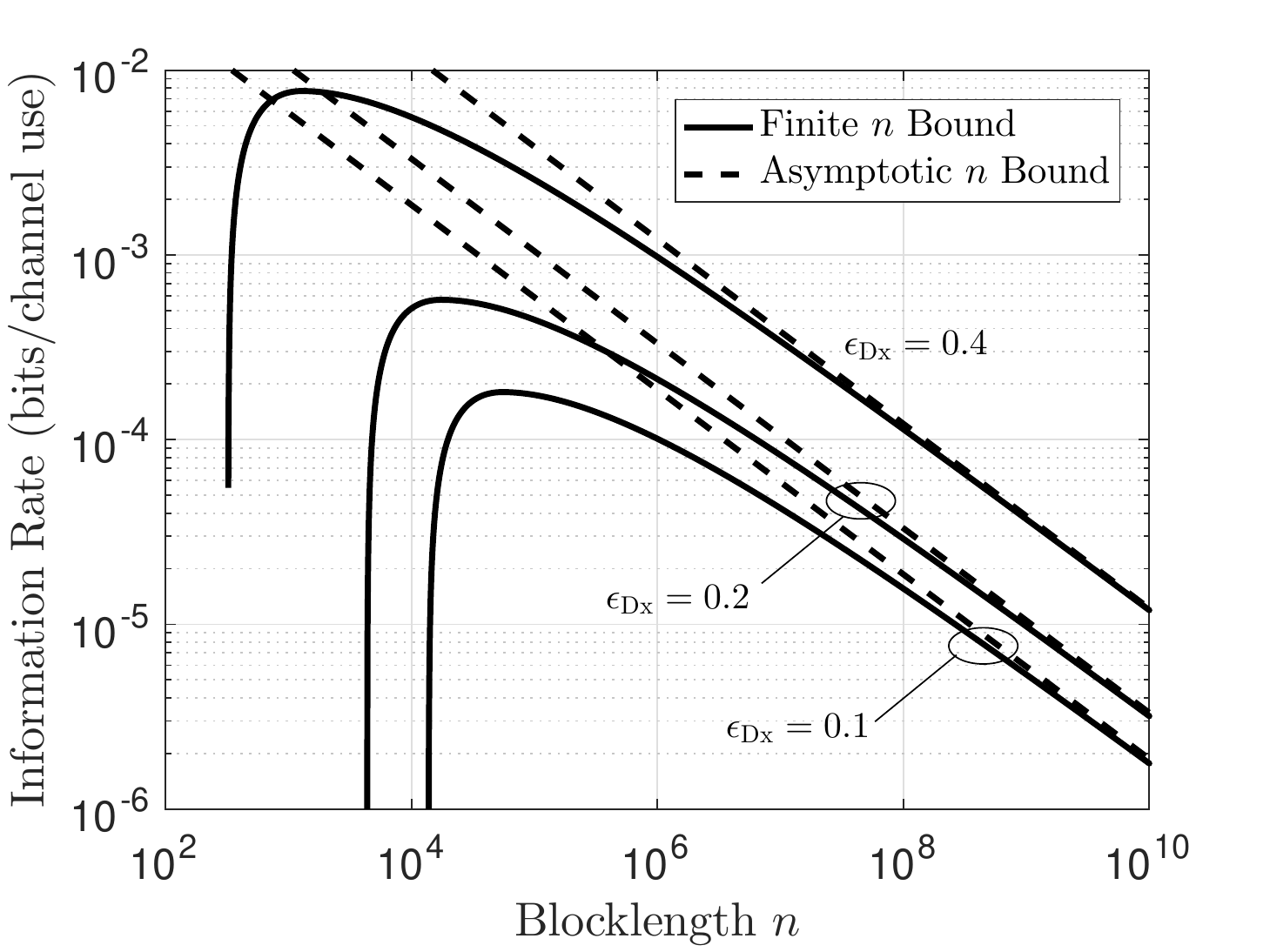} 
\caption{LPD achievable information rate (bits/channel use) for the BSC with  $\epsilon_{\rm dec} = 10^{-3}$, $\epsilon_{\rm det} = 0.1$, $\epsilon_{\rm Rx} = 0.1$ and various $\epsilon_{\rm Dx}$.  Solid lines plot~\eqref{eq:bsc_ach_bits} divided by $n$ and dashed lines plot~\eqref{eq:bsc_asymp_L} divided by $\sqrt{n}$}.
\label{fig:bsc_ach_rate}
\end{figure}
\begin{figure}[t]
\centering
\includegraphics[width=0.9\columnwidth]{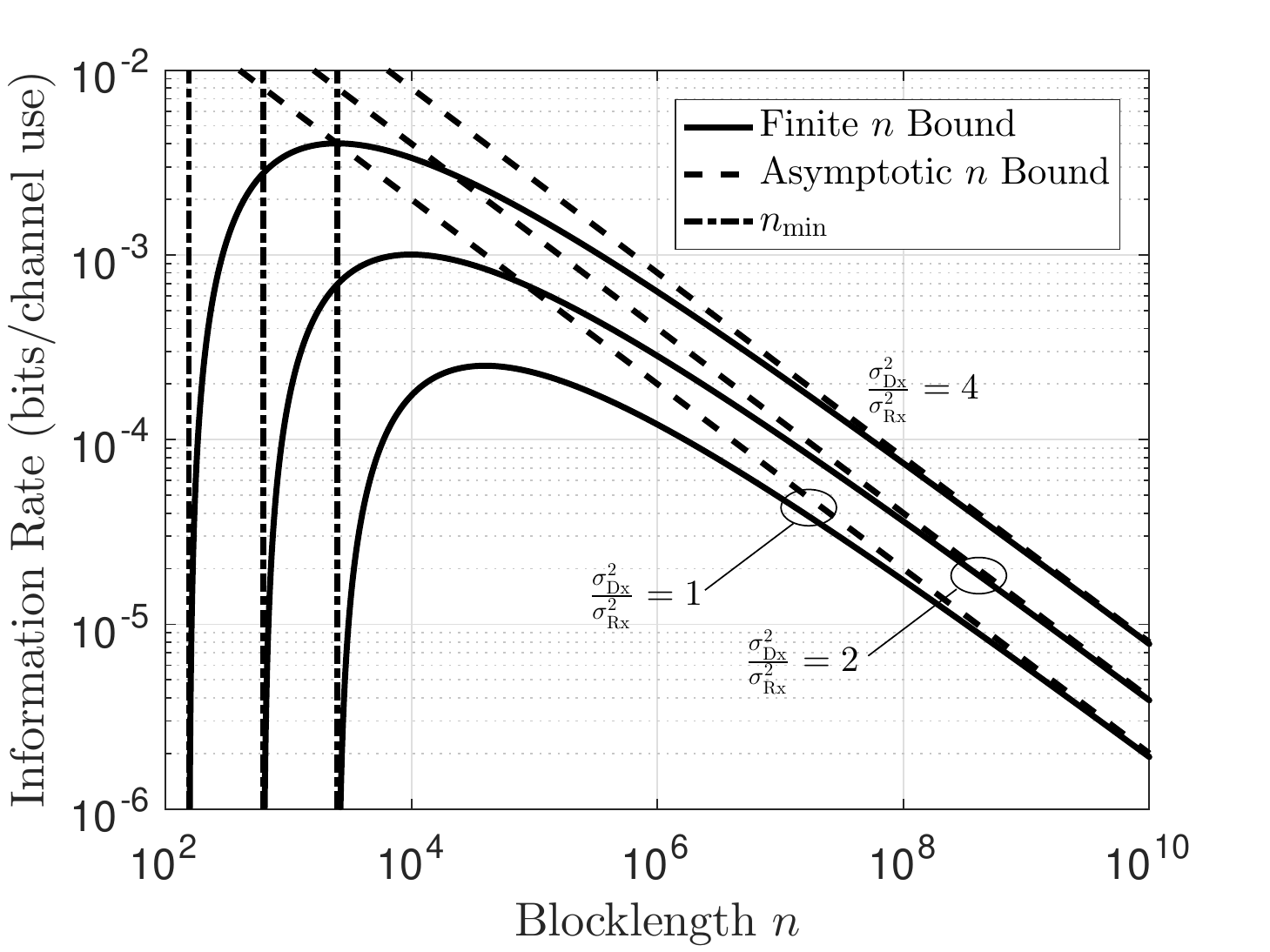} 
\caption{LPD achievable information rate (bits/channel use) for the AWGN channel with  $\epsilon_{\rm dec} = 10^{-3}$, $\epsilon_{\rm det} = 0.1$ and various noise power ratios $\sigma_{\rm Dx}^2/\sigma_{\rm Rx}^2$.  Solid lines plot~\eqref{eq:awgn_lpd} divided by $n$ and dashed lines plots the asymptotic large block length bound.  The dot-dashed lines show the minimum block length~\eqref{eq:awgn_n_min}. } \label{fig:awgn_ach_rate}
\end{figure}

%
%
\appendices

%
%
\section{Proof of Lemma~\ref{lem:sparse_exponent_minkowski}} \label{app:sparse_exponent_minkowski_proof}
Using~\eqref{eq:sparse_signalling_density} the integral in~\eqref{eq:gal_exponent} can be written as
\begin{align}
\mathcal{I} &= \int_{\mathcal{Y}} \left\{ \int_{\mathcal{X}} p_{X}(x) \left[ p_{Y |X}(y|x) \right]^{\frac{1}{1+\rho}} \, d x \right\}^{1+\rho} \, d y  \notag \\
&=  \int_{\mathcal{Y}} \left\{ (1-\tau) p^{\frac{1}{1+\rho}}(y|x_0) + \tau \int_{\mathcal{X}} p_{\tilde{X}}(x) p_{Y|X}^{\frac{1}{1+\rho}}(y|x) \, dx \right\}^{1+\rho} \, dy \notag \\
&=  \left[ \left(\int_{\mathcal{Y}} \left\{ (1-\tau) p_{Y|X}^{\frac{1}{1+\rho}}(y|x_0) + \tau \int_{\mathcal{X}} p_{\tilde{X}}(x) p_{Y|X}^{\frac{1}{1+\rho}}(y|x) \, dx \right\}^{1+\rho} \, dy \right)^{\frac{1}{1+\rho}} \right]^{1+\rho}
\end{align}
Using Minkowski's inequality,
\begin{align}
\mathcal{I} &\leq \left[  \left(\int_{\mathcal{Y}} (1-\tau)^{1+\rho} p_{Y|X}(y|x_0) \, dy \right)^{\frac{1}{1+\rho}} + \tau  \left\{ \int_{\mathcal{Y}} \left( \int_{\mathcal{X}} p_{\tilde{X}}(x) p_{Y|X}^{\frac{1}{1+\rho}}(y|x) \, dx  \right)^{1+\rho} \, dy  \right\}^{\frac{1}{1+\rho}}  \right]^{1+\rho} \notag \\
&= \left[   1-\tau \left(1 -  \left\{ \int_{\mathcal{Y}} \left( \int_{\mathcal{X}} p_{\tilde{X}}(x) p_{Y|X}^{\frac{1}{1+\rho}}(y|x) \, dx  \right)^{1+\rho} \, dy  \right\}^{\frac{1}{1+\rho}} \right)  \right]^{1+\rho} \notag \\
&= \left[   1-\tau \left(1 -  e^{- \frac{1}{1+\rho} \mathcal{E}_0(\rho,p_{\tilde{X}},p_{Y|X})} \right)  \right]^{1+\rho}
\end{align}

%
%
\section{Proof of Theorem~\ref{thrm:bsc_ach_rate}} \label{app:proof_bsc_ach_rate}
For the BSC it is not difficult to show that
\begin{align}
\mathcal{E}_{0}(\rho,p_{\tilde{X}},p_{Y|X}) &= \mathcal{E}_{0}(\rho,u,\epsilon_{\rm Rx}) = -\log \Biggl( \left[ (1-u)(1-\epsilon_{\rm Rx})^{\frac{1}{1+\rho}} + u \epsilon_{\rm Rx}^{\frac{1}{1+\rho}}  \right]^{1+\rho} \notag \\
& + \left[ (1-u)\epsilon_{\rm Rx}^{\frac{1}{1+\rho}} + u (1-\epsilon_{\rm Rx})^{\frac{1}{1+\rho}}  \right]^{1+\rho} \Biggl) \\
\chi^2( p_{\tilde{W}} \parallel p_{W|X=x_0}) &= \chi^2(u,\epsilon_{\rm Dx}) =  \frac{u^2 (1-2 \epsilon_{\rm Dx})^2}{\epsilon_{\rm Dx} (1 - \epsilon_{\rm Dx})} \label{eq:bsc_chi_dist}
\end{align}
Substituting the above expressions into~\eqref{eq:L_def} yields
\begin{align}
L_{\rm BSC}(\rho,u,\epsilon_{\rm Rx},\epsilon_{\rm Dx} ) = & \frac{2 \xi }{\log 2} \left( \frac{1+\rho}{\rho} \right) \frac{\sqrt{\epsilon_{\rm Dx}(1 - \epsilon_{\rm Dx})}}{(1 - 2 \epsilon_{\rm Dx})}  \notag \\
 \times \frac{1}{u}  \Biggl[ 1 - \Biggl\{ & \left[ (1-u)(1-\epsilon_{\rm Rx})^{\frac{1}{1+\rho}} + u\epsilon_{\rm Rx}^{\frac{1}{1+\rho}}\right]^{1+\rho} \notag \\
+ & \left[ (1-u)\epsilon_{\rm Rx}^{\frac{1}{1+\rho}} + u(1-\epsilon_{\rm Rx})^{\frac{1}{1+\rho}}\right]^{1+\rho} \Biggr\}^{\frac{1}{1+\rho}} \Biggl]. \label{eq:L_bsc2}
\end{align}
Now consider the function $g(x)$ and its derivative, i.e.
\begin{align}
g(x) &= \frac{1}{x} [ 1 - f(x)  ] \\
g'(x) &= -\frac{1}{x^2}  [ 1 - f(x)] -\frac{1}{x} f'(x)
\end{align}
where
\begin{align}
f(x)  &= \left\{ [ (1-x)a + xb]^{1+c} + [ (1-x)b + xa]^{1+c} \right\}^{\frac{1}{1+c}} \\
f'(x) &= (b-a)  f^{-c}(x) \left( [ (1-x)a + xb]^{c} - [ (1-x)b + xa]^{c} \right)
\end{align}
for $0\leq a,b,c  \leq 1$. Suppose $a \geq b$ and $0  \leq x \leq (1-b)/(a-b)$.  Then, from the above, $g'(x) < 0$ when $ [ (1-x)a + xb]^{c} \leq  [ (1-x)b + xa]^{c}$, or $x \geq \frac{1}{2}$.  Now suppose $b \geq a$ and $0  \leq x \leq (1-a)/(b-a)$. Then $g'(x) < 0$ when $ [ (1-x)a + xb]^{c} \geq  [ (1-x)b + xa]^{c}$, or $x \leq \frac{1}{2}$.  Hence,  $g(x)$ is a decreasing function for $0 \leq x \leq 1$.  Moreover, using L'Hospital's Rule
\begin{equation}
\lim_{x \rightarrow 0} g(x) = \frac{(a^c - b^c)(a-b)}{ \left(a^{1+c} + b^{1+c} \right)^{\frac{c}{1+c}}} \label{eq:fx_bsc}
\end{equation}
Hence, using $a = (1-\epsilon)^{\frac{1}{1+\rho}}$, $b = \epsilon^{\frac{1}{1+\rho}}$ and $c = \rho$, with some simplification~\eqref{eq:L_bsc2} reduces to~\eqref{eq:L_bsc} as stated. Finally, from Corollary~\ref{cor:L_asymp}, applying L'Hospital's Rule to obtain $\lim_{\rho,u \rightarrow 0} L_{\rm BSC}(\rho,u,\epsilon_{\rm Rx},\epsilon_{\rm Dx} )$ results in~\eqref{eq:bsc_asymp_L}.

%
%
\section{Proof of Theorem~\ref{thrm:awgn_ach_rate}} \label{app:proof_awgn_ach_rate}
For the AWGN channel with $p_{\tilde{X}}(x) = e^{-\frac{x^2}{2P}}/\sqrt{2 \pi P}$ computing the required integrals yields
\begin{align}
\mathcal{E}_{0}(\rho,p_{\tilde{X}},p_{Y|X}) &= \mathcal{E}_{0}(\rho,P,\sigma_{\rm Rx}^2) =  \frac{\rho}{2} \log  \left( 1+ \frac{1}{1+\rho} \frac{P}{\sigma_{\rm Rx}^2} \right) \label{eq:E0_awgn} \\
\chi^2( p_{\tilde{W}} \parallel p_{W|X=x_0}) &= \chi^2(P,\sigma_{\rm Dx}^2) = \frac{1}{\sqrt{1 - \frac{P^2}{\sigma_{\rm Dx}^4}}} - 1. \label{eq:chi_awgn}
\end{align}
Note that~\eqref{eq:E0_awgn} is equal to~\cite[eq.~(124)]{Gallager1965p3} with $r = 0$, since the LPD constraint is accounted for by~\eqref{eq:tau_constraint}. Substituting the above above expressions into~\eqref{eq:L_def} yields
\begin{align}
L_{\rm AWGN} (\rho,P,\sigma_{\rm Rx}^2, \sigma^2_{\rm Dx} ) = \frac{2\xi }{\log 2} \left(\frac{1+\rho}{\rho} \right) & \left[ \frac{1}{\sqrt{1 - \frac{P^2}{\sigma_{\rm Dx}^4}}} - 1 \right]^{-\frac{1}{2}} \notag \\
 \times & \left[ 1 - \left(1 + \frac{1}{1+\rho} \frac{P}{\sigma_{\rm Rx}^2} \right)^{-\frac{\rho}{2(1+\rho)}} \right].  \label{eq:L_awgn2}
\end{align}
Now consider the function
\begin{equation}
f(x) = \frac{ 1 - (1+ax)^{-c}}{\sqrt{\frac{1}{\sqrt{1 - x^2}}-1}} = \sqrt{2} ac  \left( \frac{1 - \frac{1}{2} c(c+1) a^2 x + \mathcal{O}(x^2) }{ 1 + \frac{3}{16} x^2 + \mathcal{O}(x^4)} \right)
\end{equation}
for $a,c > 0$. It is clear from the expansion of $f(x)$ that it is a decreasing function for $0 < x < 1$. Moreover,
\begin{equation}
\lim_{x \rightarrow 0} f(x) = \sqrt{2} a c. \label{eq:fx_awgn}
\end{equation}
Therefore, using $a = \frac{1}{1+\rho} \frac{\sigma_{\rm Dx}^2}{\sigma_{\rm Rx}^2}$ and $c = \frac{\rho}{2 (1+\rho)}$ in~\eqref{eq:fx_awgn} results in
\begin{align}
\lim_{P \rightarrow 0} L_{\rm AWGN} (\rho,P,\sigma_{\rm Rx}^2, \sigma^2_{\rm Dx} ) = \frac{\sqrt{2} \xi }{(1+\rho)\log 2} \frac{\sigma_{\rm Dx}^2}{\sigma_{\rm Rx}^2}.  \label{eq:L_awgn}
\end{align}
Hence substituting~\eqref{eq:L_awgn} into~\eqref{eq:ach_covert_bits} yields~\eqref{eq:awgn_lpd} as stated.  Now consider the function $g(x) = \frac{a}{1+x} - \frac{b}{x} $ for $a  > b > 0$.  It is not difficult to show that $g(x)$ is maximised when $ x = \frac{b}{a-b} \left(1 + \sqrt{\frac{a}{b}} \right)$.  Applying this solution to~\eqref{eq:awgn_lpd} yields the optimal value for $\rho$, i.e.~\eqref{eq:rho_opt_awgn}.  Applying Corollary~\ref{cor:optimal_n} and maximising over $\rho$ yields~\eqref{eq:R_star_awgn} and~\eqref{eq:n_star_awgn}. Finally~\eqref{eq:k_star_awgn} is simply the product of~\eqref{eq:R_star_awgn} and~\eqref{eq:n_star_awgn}.

\bibliographystyle{IEEE}


\end{document}